\documentclass[journal,draftcls,onecolumn,12pt,twoside]{IEEEtran}

\usepackage{cite}

\ifCLASSINFOpdf
\else
\fi

\usepackage[cmex10]{amsmath}
\usepackage{amsthm}
\usepackage{amssymb}
\usepackage{bm}
\usepackage{latexsym}
\usepackage{bbm}
\usepackage{subfigure}
\usepackage{color}
\usepackage[dvipdfmx]{graphicx}

\newtheorem{theorem}{Theorem}
\newtheorem{corollary}{Corollary}
\newtheorem{lemma}{Lemma}
\newtheorem{remark}{Remark}

\newcommand{\E}{\mathrm{E}}

\newcommand{\peak}{\mathrm{peak}}
\newcommand{\dd}{\mathrm{d}}

\newcommand{\one}{\mathbbm{1}}
\newcommand{\cls}[1]{\langle #1 \rangle}
\newcommand{\all}{\mathrm{all}}

\allowdisplaybreaks

\begin{document}

\title{Exact Analysis of the Age of Information in the Multi-Source
M/GI/1 Queueing System}

\author{Yoshiaki Inoue and Tetsuya Takine
\thanks{This work was supported in part by JSPS KAKENHI Grant Number
21H03399. Y. Inoue and T. Takine are with 
Department of Information and Communications Technology, 
Graduate School of Engineering, Osaka University, Suita 565-0871, Japan.

E-mail:\{yoshiaki, takine\}@comm.eng.osaka-u.ac.jp.}}%
%

\maketitle
\IEEEpeerreviewmaketitle

\begin{abstract}
We consider a situation that multiple monitoring applications (each
with a different sensor-monitor pair) compete for 
a common service resource such as a communication link.
Each sensor reports the latest state of its own time-varying
information source to its corresponding monitor, incurring queueing
and processing delays at the shared resource. 
The primary performance metric of interest is the age of information
(AoI) of each sensor-monitor pair, which is defined as
the elapsed time 
from the generation of the information currently
displayed on the monitor.
Although the multi-source first-come first-served (FCFS) M/GI/1
queue is one of the most fundamental model to describe such competing
sensors, its exact analysis has been an open problem for years. 
In this paper, we show that the Laplace-Stieltjes
transform (LST) of the stationary distribution of the AoI in this
model, as well as the mean AoI, is given by 
a simple explicit formula, utilizing the double Laplace transform
of the transient workload in the M/GI/1 queue.
\end{abstract}

\section{Introduction}

We consider a situation that multiple time-varying information sources
are monitored remotely. Specifically, each information source is
attached with a sensor node, which reports observed data to a remote
monitor that displays the latest state information received. 
Each observed data needs to be ``served'' by an intermediate service
resource (such as a communication link and a processor) before it is
received by the monitor, so that it experiences queueing and
processing delay there. The service resource is shared by the
source-monitor pairs in the system and the degree of congestion is
highly affected by the sampling rate of the sensors.

Because the state of the information source changes over time, the
value of displayed information degenerates with time. It is thus
important for such a system that the information displayed on each
monitor is kept sufficiently fresh.
The age of information (AoI) \cite{Kaul12-1,Yates21} is a widely used metric
of the freshness of information, which is defined as the elapsed time
from the generation time of currently displayed information. 
Specifically, letting $K$ denote the number of sensor-monitor pairs,
the AoI $A_t^{\cls{k}}$ of the $k$th ($k=0,1,\ldots,K-1$)
sensor-monitor pair at time $t$ is defined as
\begin{equation}
A_t^{\cls{k}} = t - \eta_t^{\cls{k}},
\quad
t \in \mathbb{R},
\label{eq:A_t-def}
\end{equation}
where $\eta_t^{\cls{k}}$ denotes the generation time of the displayed
information at time $t$. 

The mean AoI $\E[A^{\cls{k}}]$ is the primary
performance measure of interest, which is defined as
\[
\E[A^{\cls{k}}] = \lim_{T \to \infty}\frac{1}{T} \int_{-T/2}^{T/2}
A_t^{\cls{k}} \dd t.
\]
Under a fairly general setting, the mean AoI $\E[A^{\cls{k}}]$ is given in terms
of the mean inter-sampling time, the mean delay, and a cross-term
of the inter-sampling time and the delay \cite{Kaul12-1}.
A more detailed characterization of the AoI process is provided by 
the probability distribution $F_{A^{\cls{k}}}(x)$ ($x \geq 0$) of the
AoI, which is defined as a long-run fraction of time
that the AoI process does not exceed $x$:
\[
F_{A^{\cls{k}}}(x)
= 
\lim_{T \to \infty}\frac{1}{T} \int_{-T/2}^{T/2}
\one\{A_t^{\cls{k}} \leq x\} \dd t,
\]
which is known to be given in terms of the probability distributions of
the system delay and the peak AoI value just before updates \cite{Inoue18}.

To characterize the AoI process formally, the system is usually modeled as a
queueing system, where generation of information packets by sensors corresponds
to ``arrivals'', and their reception by monitors corresponds to
``departures''.
Analytical results for the AoI in various queueing systems have been reported in the literature,
e.g., first-come first-served (FCFS) queues \cite{Kaul12-1,Inoue18}, 
last-come first-served (LCFS) queues \cite{Kaul12-2,Costa16,Champati19}, 
and multi-server queues \cite{Kam16,Inoue2021}.
The FCFS queueing model represents the situation where the
application layer does not have any control over the data transmission
performed by lower layer protocols, so that the AoI performance in
this framework is interpreted as that of monitoring applications implemented on top of
the conventional communication mechanisms. 
The LCFS queueing model, on the other hand, gives a priority to newer
information in transmission, which in general leads to better performance 
at the expense of implementation/operation costs, especially due to
the necessity of the cross-layer design and intelligent
functionalities of network components.

The models mentioned above \cite{Kaul12-1,Inoue18, Kaul12-2,Costa16,Champati19,Kam16} 
assume that there is only a single sensor node in the system (i.e., $K=1$), so that
the ``congestion'' of the system refers to the accumulation of
buffered packets generated by the sensor itself.
In order to incorporate the existence of competing sensors that share
the communication resources, it is necessary to consider
\textit{multi-source} queueing models, where each
source corresponds to a sensor node.
In fact, significant research efforts have recently been put into generalizing the AoI
analysis to multi-source queueing models.
Such generalization has been particularly successful for the multi-source
\textit{buffer-less} M/G/1/1 
queues \cite{Dogan2021,Miyoshi2022,Chen2022,Moltafet2022,Abd-Elmagid2023,Zhang2023},
whose analysis is less complicated due to the small state space.

For the standard (infinite-buffer) multi-source FCFS systems, however, the exact analysis of the AoI
has been reported for the simplest M/M/1 queueing model only \cite{Yates2019, Moltafet2020}.
The difficulty in the analysis of the multi-source FCFS
M/GI/1 queue is explained in \cite{Moltafet2020}, where approximate formulas
for the mean AoI are proposed. Roughly speaking, in multi-class FCFS queues, 
an unbounded number of information packets from other sources can arrive
between two successive packets of a single source, which makes it necessary to 
keep track of the whole dynamics of the system driven by other sources.
Because of such difficulty, exact analysis of the AoI in the
multi-source FCFS M/GI/1 queue has been an open problem for years.

The purpose of this paper is to provide a solution to this problem,
by showing that the Laplace-Stieltjes transform (LST) of the stationary
distribution of the AoI in the multi-source FCFS M/GI/1 queue, as well
as the mean AoI, is given by a simple explicit formula.
The key observation is that the system dynamics between 
two successive packets from a single source is characterized in terms
of the double transform \cite[P.\ 83]{Takagi1991} of the transient
workload process in the M/GI/1 queue.
As will be shown later, the complexity of the exact mean AoI formula 
is more or less the same as that of the single-source M/GI/1 queue.

The rest of this paper is organized as follows.
We first provide a formal description of the model in Section \ref{sec:model}. 
We then present the exact analysis of the AoI in Section
\ref{sec:analysis}, where explicit formulas of the LST and the mean 
of the AoI are derived.
Finally, we conclude this paper in Section \ref{sec:conclusion}.

\section{Model}\label{sec:model}

Throughout this paper, we comply with the following rules of notation:
for any non-negative random variable $Y$, its cumulative distribution function (CDF)
is denoted by $F_Y(x)$ ($x \geq 0$) and its LST is denoted by
$f_Y^*(s)$ ($s > 0$):
\[
F_Y(x) := \Pr(Y \leq x),
\qquad
f_Y^*(s) := \E[e^{-sY}] = \int_0^{\infty} e^{-sx} \dd F_Y(x).
\]
We also apply this rule in the opposite direction: given a CDF denoted
by $F_Y(x)$ ($x \geq 0$), the corresponding random variable $Y$ is
defined accordingly.

We consider an FCFS single-server queue with
$K$ different source-monitor pairs. The $k$th ($k=0,1,\ldots,K-1$)
source generates 
information packets according to a Poisson process
with rate $\lambda^{\cls{k}}$ ($\lambda^{\cls{k}} > 0$), 
which arrive at the server immediately after their generations. 
Packets are served by the server in
order of arrival and they update the corresponding monitor on
completion of the service. 
Hereafter we shall refer to the $k$th source and monitor as source
$k$ and monitor $k$ for short.
We assume that service times of packets from 
source $k$ are i.i.d.\ according to a general non-negative
probability distribution with CDF $F_{H^{\cls{k}}}(x)$ ($x \geq 0$),
i.e., the service time distribution
is allowed to differ depending on the source-monitor pair.
Let $\rho^{\cls{k}}$ denote the traffic intensity of source $k$:
\[
\rho^{\cls{k}} = \lambda^{\cls{k}} \E[H^{\cls{k}}],
\quad
k = 0,1,\ldots,K-1.
\]
We assume $\sum_{k=1}^{K} \rho^{\cls{k}} < 1$, so that the system is
stable. Also, we assume that the system is stationary unless otherwise
mentioned.

To define the AoI process for each monitor,
let $\alpha_n^{\cls{k}}$ ($n \in \mathbb{Z}$) denote the generation
time of the $n$th packet of source $k$. Without loss of
generality, we set packet indices so that $\alpha_{-1}^{\cls{k}} < 0 \leq
\alpha_0^{\cls{k}}$ holds for each $k$.
Let $\beta_n^{\cls{k}}$ ($n \in \mathbb{Z}$) denote the reception time of the
$n$th packet by monitor $k$. Its system delay $D_n^{\cls{k}}$ ($n \in \mathbb{Z}$) 
is then given by $D_n^{\cls{k}} = \beta_n^{\cls{k}} - \alpha_n^{\cls{k}}
$.
 
The AoI $A_t^{\cls{k}}$ ($t \geq 0$, $k=1,2,\ldots,K$) of monitor
$k$ at time $t$ is then defined as
\[
A_t^{\cls{k}} = t - \alpha_n^{\cls{k}},
\quad
\beta_n^{\cls{k}} \leq t < \beta_{n+1}^{\cls{k}}.
\]
%
Let $A_{\peak,n}^{\cls{k}}$ denote 
the $n$th peak AoI of monitor
$k$, i.e., the value of AoI just before $n$th update of the monitor:  
\begin{align}
A_{\peak,n}^{\cls{k}}
= 
\lim_{t \to \beta_n^{\cls{k}}-} A_t^{\cls{k}}
&=
\beta_n^{\cls{k}} - \alpha_{n-1}^{\cls{k}}
\nonumber
\\
&=
\beta_n^{\cls{k}} - \alpha_n^{\cls{k}}
+ \alpha_n^{\cls{k}} - \alpha_{n-1}^{\cls{k}}
\nonumber
\\
&=
D_n^{\cls{k}} + G_n^{\cls{k}},
\label{eq:A_peak-DG}
\end{align}
where $G_n^{\cls{k}} := \alpha_n^{\cls{k}} - \alpha_{n-1}^{\cls{k}}$
denotes the inter-generation time between $n$th and $(n-1)$st packets
of source $k$.

Because of the stationarity of the system, the 
probability distribution of $A_t^{\cls{k}}$ does not depend on $t$.
Let $A^{\cls{k}}$ denote a generic random variable following the
stationary distribution of $A_t^{\cls{k}}$.
Similarly, let $G^{\cls{k}}$, $D^{\cls{k}}$, and $A_{\peak}^{\cls{k}}$ denote generic
random variables following the stationary distributions of 
$G_n^{\cls{k}}$, $D_n^{\cls{k}}$ and $A_{\peak,n}^{\cls{k}}$.

\section{Exact Analysis of the AoI distribution}
\label{sec:analysis}

For better readability, we mainly focus on the AoI $A^{\cls{0}}$ of
source $0$ in the analysis below, which can be readily done without loss of generality. 

Because the sequences of amounts of
work brought by sources $1,2,\ldots,K-1$ (excluding source $0$)
are all represented as compound Poisson processes, 
their superposition also forms a compound Poisson process with rate
$\lambda^+$ and the random jump size $H^+$, where 
\begin{align*}
\lambda^+ = \sum_{k=1}^{K-1} \lambda^{\cls{k}},
\quad
f_{H^+}^*(s)
=
\sum_{k=1}^{K-1}
\frac{\lambda^{\cls{k}}}{\lambda^+} \cdot f_{H^{\cls{k}}}^*(s).
\end{align*}
Similarly, we define $\lambda^{\all}$ and $H^{\all}$ as the 
rate and random jump size of the compound Poisson process composed of
all sources including source $0$:
\begin{align}
\lambda^{\all} = \lambda^{\cls{0}} + \lambda^+,
\quad
f_{H^{\all}}^*(s) 
= 
\frac{
\lambda^{\cls{0}} f_{H^{\cls{0}}}^*(s)
+
\lambda^+ f_{H^+}^*(s)
}{\lambda^{\all}}.
\label{eq:lmd-H-all-def}
\end{align}
We further define $\rho^{\all}$ as the total traffic intensity:
\begin{equation}
\rho^{\all} = \rho^{\cls{0}} + \rho^+,
\label{eq:rho-all-def}
\end{equation}
where 
\[
\rho^+ = \rho_1+\rho_2+\cdots+\rho_{K-1}.
\]

To make the notation even simpler, we hereafter drop the superscript
$\cls{0}$ for quantities of source $0$. For example, the inter-generation
time $G_n^{\cls{0}}$ and system delay $D_n^{\cls{0}}$ of source $0$
are written simply as
\begin{align*}
G_n := G_n^{\cls{0}},
\quad
D_n := D_n^{\cls{0}}.
\end{align*}
Similarly, the AoI and peak AoI of source $0$ are written as
\[
A_t := A_t^{\cls{0}},
\quad
A_{\peak,n} = A_{\peak,n}^{\cls{0}}. 
\]
The same rules apply to their stationary versions:
\[
G := G^{\cls{0}},
\quad
D := D^{\cls{0}},
\quad
A := A^{\cls{0}}.
\]
Notice that (\ref{eq:lmd-H-all-def}) and (\ref{eq:rho-all-def}) are rewritten as
\[
\lambda^{\all} = \lambda + \lambda^+,
\qquad
f_{H^{\all}}^*(s) 
= 
\frac{
\lambda f_{H}^*(s)
+
\lambda^+ f_{H^+}^*(s)
}{\lambda^{\all}},
\qquad
\rho^{\all} = \rho + \rho^+.
\]

The starting point of our analysis is the following relation known in
the literature \cite{Inoue18}:
\begin{lemma}
\label{lemma:general-formula}
The LST of the stationary AoI $f_A^*(s)$ of source $0$ is represented in
terms of the LSTs of its stationary system delay $f_D^*(s)$ 
and stationary peak AoI $f_{A_{\peak}}^*(s)$ as
\[
f_A^*(s) 
= 
\lambda
\cdot
\frac{f_D^*(s) - f_{A_{\peak}}^*(s)}{s},
\quad
s > 0.
\]
\end{lemma}
Because the stationary waiting time of packets of source $0$ is identical to that
in the single-input M/GI/1 queue with the arrival rate $\lambda^{\all}$ and 
service time LST $f_{H^{\all}}^*(s)$, we have
\begin{equation}
f_W^*(s) = \frac{(1-\rho^{\all})s}{s - \lambda^{\all} + \lambda^{\all} f_{H^{\all}}^*(s)}
=
\frac{(1-\rho-\rho^+)s}{s - \lambda - \lambda^+ + \lambda f_H^*(s) + \lambda^+ f_{H^+}^*(s)}.
\end{equation}
It then follows from $f_D^*(s)=f_W^*(s)f_H^*(s)$ that 
\begin{align}
f_D^*(s) 
&=
\frac{(1-\rho - \rho^+)s}{s - \lambda - \lambda^+ +  \lambda f_H^*(s)
+ \lambda^+ f_{H^+}^*(s)} \cdot f_{H}^*(s).
\label{eq:f_D}
\end{align}
We thus consider $f_{A_{\peak}}^*(s)$ below.
\begin{lemma}
\label{lemma:A_peak}

The LST $f_{A_{\peak}}^*(s)$ of the stationary peak AoI of source $0$
is given by
\begin{equation}
f_{A_{\peak}}^*(s)
=
\frac{\lambda f_H^*(s)}{s + \lambda - \phi(s)}
\left\{
f_D^*(s)
-
\frac{sf_D^*\left(\psi(s+\lambda)\right)}{\psi(s+\lambda)}
\right\},
\label{eq:A_peak-LST}
\end{equation}
where $\phi(s)$ ($s > 0$) is defined as
\begin{equation}
\phi(s) = s -\lambda^+ + \lambda^+ f_{H^+}^*(s),
\label{eq:phi-def}
\end{equation}
and $\psi(\omega)$ ($\omega > 0$) denotes the inverse function of $\phi(s)$, i.e.,
\begin{equation}
\phi(\psi(\omega)) = \omega.
\label{eq:psi-def}
\end{equation}
\end{lemma}
\begin{remark}
\label{remark:psi}
It is readily verified that $\phi(s)$ is an increasing continuous function for $s > 0$:
\[
\frac{\dd \phi(s)}{\dd s} = 1 + \lambda^+ \cdot \frac{\dd}{\dd s}\E[e^{-sH^+}]
=
1 - \rho^+ \E[e^{-sH^+}] > 0.
\]
Also, we have $\lim_{s \to 0+}\phi(s) = 0$ and $\lim_{s \to
\infty}\phi(s) = \infty$, which ensures the existence of the inverse
function $\psi(\omega)$ ($\omega > 0$).
\end{remark}
\begin{proof}
From (\ref{eq:A_peak-DG}), we have
\[
A_{\peak,n+1}
= 
G_{n+1} + D_{n+1}
=
G_{n+1} + W_{n+1} + H_{n+1},
\]
where $W_n$ denotes the waiting time of the $n$th packet of source
$0$. We then have
\begin{align}
f_{A_{\peak}}^*(s)
&=
\E\left[e^{-sA_{\peak,n+1}}\right]
\nonumber
\\
&=
\E\left[
e^{-s(G_{n+1}+W_{n+1}+H_{n+1})}
\right]
\nonumber
\\
&=
f_H^*(s)
\cdot
\E\left[
e^{-s(G_{n+1}+W_{n+1})}
\right]
\nonumber
\\
&=
f_H^*(s)
\int_{x=0}^{\infty}
\dd F_{D_n}(x)
\int_{y=0}^{\infty}
\E\left[
e^{-s(y+W_{n+1})}
\mid
D_n = x, G_n = y
\right]
\lambda e^{-\lambda y} \dd y
\nonumber
\\
&=
f_H^*(s)
\int_{x=0}^{\infty}
\dd F_D(x)
\int_{y=0}^{\infty}
\E\left[
e^{-sW_{n+1}} \mid D_n = x, G_n = y
\right]
\lambda e^{-(s+\lambda) y} \dd y.
\label{eq:A-peak-calc0}
\end{align}
Note here that the waiting time $W_{n+1}$ of the ($n+1$)st packet 
of source $0$ is given in terms of the transient workload $V_t$ of
the M/GI/1 queue with arrival rate $\lambda^+$ and the service time
distribution $F_{H^+} (x)$. More specifically, its conditional LST
given $D_n$ and $G_n$ is given by
\begin{align*}
\E\left[
e^{-sW_{n+1}}
\mid
D_n = x, G_n = y
\right]
=
\E[e^{-sV_y} \mid V_0 = x].
\end{align*}
It is known that the double-transform of the workload
process $V_t$ in the M/GI/1 queue is given by \cite[P.\ 83]{Takagi1991}:
\[
\int_0^{\infty}
\E[e^{-sV_y} \mid V_0 = x]
e^{-\omega y}
\dd y
=
\frac{1}{\omega - \phi(s)} 
\left\{
e^{-sx}
-
\frac{se^{-\psi(\omega)x}}{\psi(\omega)}
\right\},
\]
where $\phi(s)$ and $\psi(\omega)$ are defined as (\ref{eq:phi-def})
and (\ref{eq:psi-def}).

Therefore, we have from (\ref{eq:A-peak-calc0}),
\begin{align*}
f_{A_{\peak}}^*(s)
&=
\lambda 
f_H^*(s)
\int_{x=0}^{\infty}
\dd F_D(x)
\int_{y=0}^{\infty}
\E[e^{-sV_y} \mid V_0 = x]
e^{-(s+\lambda) y} \dd y
\\
&=
\lambda 
f_H^*(s)
\int_{x=0}^{\infty}
\frac{1}{s + \lambda - \phi(s)} 
\left\{
e^{-sx}
-
\frac{se^{-\psi(s+\lambda)x}}{\psi(s+\lambda)}
\right\}
\dd F_D(x),
\end{align*}
which implies (\ref{eq:A_peak-LST}).
\end{proof}

We then obtain the following expression for the LST of the AoI
distribution:
\begin{theorem}
\label{thm:f_A}
The LST $f_A^*(s)$ of the stationary AoI $A$ is given by
\begin{equation}
f_A^*(s)
=
\frac{\lambda f_H^*(s)}{s + \lambda - \phi(s)}
\left\{
f_W^*(s) - (1-\rho-\rho^+)
+
\frac{\lambda f_D^*\left(\psi(s+\lambda)\right)}{\psi(s+\lambda)}
\right\},
\label{eq:f_A-final}
\end{equation}
where $f_W^*(s)$ denotes the LST of the stationary waiting time $W$:
\[
f_W^*(s) 
= 
\frac{(1-\rho - \rho^+)s}{\phi(s) - \lambda +  \lambda f_H^*(s)}.
\]
\end{theorem}
\begin{proof}
Using $\phi(s)$ defined in (\ref{eq:phi-def}), we rewrite 
the expression for the LST $f_D^*(s)$ of the stationary system delay as
\begin{equation}
f_D^*(s) 
= 
\frac{(1-\rho - \rho^+)s}{\phi(s) - \lambda +  \lambda f_H^*(s)} \cdot
f_{H}^*(s).
\label{eq:f_D-by-phi}
\end{equation}
We then have from Lemma \ref{lemma:general-formula} and Lemma \ref{lemma:A_peak},
\begin{align*}
f_A^*(s)
&=
\frac{\lambda}{s}
\left\{
f_D^*(s) - f_{A_{\peak}}^*(s) 
\right\}
\\
&=
\frac{\lambda}{s}
\cdot
\frac{1}{s + \lambda - \phi(s)}
\left\{
(s+\lambda -\phi(s) - \lambda f_H^*(s)) f_D^*(s)
+
\frac{s\lambda f_H^*(s)f_D^*\left(\psi(s+\lambda)\right)}{\psi(s+\lambda)}
\right\}
\\
&=
\frac{\lambda}{s}
\cdot
\frac{1}{s + \lambda - \phi(s)}
\left\{
sf_D^*(s) - (1-\rho-\rho^+)s \cdot f_H^*(s)
+
\frac{s\lambda f_H^*(s)f_D^*\left(\psi(s+\lambda)\right)}{\psi(s+\lambda)}
\right\},
\end{align*}
which implies (\ref{eq:f_A-final}).
\end{proof}
Taking the derivative of (\ref{eq:f_D}) and letting $s \to 0+$, we
have the mean system delay $\E[D]$: 
\begin{equation}
\E[D] = \frac{\lambda\E[H^2]+\lambda^+\E[(H^+)^2]}{2(1-\rho-\rho^+)} + \E[H].
\label{eq:ED}
\end{equation}
An explicit formula for the mean AoI is then obtained as follows:
\begin{corollary}
\label{cor:EA}
The mean AoI $\E[A]$ is given by
\begin{equation}
\E[A] 
= 
\frac{\lambda\E[H^2]+\lambda^+\E[(H^+)^2]}{2(1-\rho-\rho^+)} + \E[H]
+
\frac{\rho^+}{\lambda}
+ 
\frac{1-\rho - \rho^+}{\lambda f_H^*(\gamma)},
\label{eq:EA}
\end{equation}
where $\gamma$ is defined as
\[
\gamma := \psi(\lambda),
\]
i.e., it represents the unique solution of the following equation:
\begin{equation}
x - \lambda -\lambda^+ + \lambda^+ f_{H^+}^*(x) = 0, \quad
\lambda \leq x \leq \lambda + \lambda^+.
\label{eq:gamma-def}
\end{equation}
\end{corollary}
The proof of Corollary \ref{cor:EA} is 
straightforward; we provide the proof in the appendix. 

The expression (\ref{eq:EA}) for the mean AoI $\E[A]$ largely
resembles the known result for the single-source M/GI/1 queue.
In fact, it is readily verified that as $\lambda^+ \to 0+$, we have
$\gamma \to \lambda$ and
\[
\E[A] \to \frac{\lambda\E[H^2]}{2(1-\rho)} + \E[H] + \frac{1-\rho}{\lambda f_H^*(\lambda)},
\]
which agrees with the mean AoI in the single-source M/GI/1 queue \cite[Eq.\ (36)]{Inoue18}.

\section{Conclusion}\label{sec:conclusion}

In this paper, we presented the exact analysis of the AoI in the
multi-source FCFS M/GI/1 queueing model. We derived the explicit formula for
the LST of the stationary distribution of the AoI in Theorem
\ref{thm:f_A}, based on the observation that the LST of the stationary
peak AoI distribution is tractable using the explicit double-transform
of the transient workload in the M/GI/1 queue.
We also showed in Corollary \ref{cor:EA} that the mean AoI $\E[A]$ in
this model is given by a simple explicit formula. 
The result shows that the complexity of the mean AoI formula in the multi-source M/GI/1 queue 
is more or less the same as that of the single-source case, suggesting
its usefulness as a simple framework in understanding the AoI
performance in multi-source communication systems.

\appendix

Note first that the existence and uniqueness of the solution of
(\ref{eq:gamma-def}) is verified in the same way as Remark
\ref{remark:psi}: the left-hand side of this equation is increasing
and continuous for $\lambda \leq x \leq \lambda+\lambda^+ $ and it
takes a negative value at $x = \lambda$ and a positive value at $x = \lambda+\lambda^+$.

We decompose the LST of the AoI into the following product form
(\ref{eq:f_A-final}):
\[
f_A^*(s) = f_{A,1}^*(s) \cdot f_{A,2}^*(s),
\]
where
\begin{align}
f_{A,1}^*(s)
&:=
\frac{\lambda f_H^*(s)}{s + \lambda - \phi(s)},
\label{eq:f_A1-def}
\\
f_{A,2}^*(s)
&:=
f_W^*(s) - (1-\rho-\rho^+)
+
\frac{\lambda f_D^*\left(\psi(s+\lambda)\right)}{\psi(s+\lambda)}.
\label{eq:f_A2-def}
\end{align}
Noting that the definition of $\gamma$ 
and (\ref{eq:gamma-def}) imply $\phi(\gamma) = \lambda$,
we have from (\ref{eq:f_D-by-phi}),
\begin{equation}
f_D^*(\gamma) = \frac{(1-\rho-\rho^+)\gamma}{\lambda}.
\label{eq:f_D-gamma}
\end{equation}
It then follows from $\lim_{s \to 0}\phi(s)= 0$ that
\[
\lim_{s \to 0+} f_{A,1}^*(s) = 1,
\quad
\lim_{s \to 0+} f_{A,2}^*(s) = 1.
\]
Therefore, letting
\[
f_{A,i}^{(1)}(s) := (-1)\cdot \frac{\dd f_{A,i}^*(s)}{\dd s},
\quad
i = 1,2,
\]
we have
\begin{align}
\E[A] 
&= 
(-1)\cdot\lim_{s \to 0+}\frac{\dd f_A^*(s)}{\dd s}
\nonumber
\\
&=
\lim_{s \to 0+}
f_{A,1}^{(1)}(s)
+
\lim_{s \to 0+}
f_{A,2}^{(1)}(s).
\label{eq:EA-by-A1-A2}
\end{align}
We thus consider $f_{A,1}^{(1)}(s)$ and $f_{A,2}^{(1)}(s)$ below.

By definition (cf.\ (\ref{eq:phi-def}) and (\ref{eq:psi-def})), we have
\[
\frac{\dd \phi(s)}{\dd s}
= 
1 - \lambda^+ f_{H^+}^{(1)}(s),
\quad
\frac{\dd \psi(\omega)}{\dd \omega}
= 
\frac{1}{1 - \lambda^+ f_{H^+}^{(1)}(\psi(\omega))},
\]
where
\[
f_{H^+}^{(1)}(s) := (-1)\cdot \frac{\dd f_{H^+}^*(s)}{\dd s}.
\]
It then follows immediately from
(\ref{eq:f_A1-def}), (\ref{eq:f_A2-def}), and (\ref{eq:f_D-gamma})
that
\begin{align}
\lim_{s \to 0+}
f_{A,1}^{(1)}(s)
&=
\E[H] + \frac{\rho^+}{\lambda},
\label{eq:f_A1}
\\
\lim_{s \to 0+}
f_{A,2}^{(1)}(s)
&=
\E[W] 
+
\frac{\lambda f_D^{(1)}(\psi(\lambda))}{\psi(\lambda)}
\cdot
\lim_{\omega \to \lambda} \frac{\dd \psi(\omega)}{\dd \omega}
+
\frac{\lambda f_D^*(\psi(\lambda))}{(\psi(\lambda))^2}
\cdot
\lim_{\omega \to \lambda} \frac{\dd \psi(\omega)}{\dd \omega}
\nonumber
\\
&=
\E[W] 
+ 
\frac{1}{1 - \lambda^+ f_{H^+}^{(1)}(\gamma)}
\left(
\frac{\lambda}{\gamma} \cdot f_D^{(1)}(\gamma)
+
\frac{\lambda f_D^*(\gamma)}{\gamma^2}
\right)
\nonumber
\\
&=
\E[W] 
+ 
\frac{1}{\{1 - \lambda^+ f_{H^+}^{(1)}(\gamma)\}\gamma}
\left(
\lambda f_D^{(1)}(\gamma)
+
1-\rho-\rho^+
\right),
\label{eq:f_A2}
\end{align}
where
\[
f_D^{(1)}(s) := (-1)\cdot \frac{\dd f_D^*(s)}{\dd s}.
\]
%

The rest 
is to determine $f_D^{(1)}(\gamma)$, which is
straightforward with 
$\phi(\gamma) = \lambda$:  
\begin{align*}
f_D^{(1)}(\gamma) 
&= 
-
\frac{1-\rho - \rho^+}{\phi(\gamma) - \lambda +  \lambda f_H^*(\gamma)} \cdot
f_{H}^*(\gamma)
\\
&\quad{}+
\frac{(1-\rho - \rho^+)\gamma}
{(\phi(\gamma) - \lambda +  \lambda f_H^*(\gamma))^2} 
\cdot
\{1-\lambda^+ f_{H^+}^{(1)}(\gamma) - \lambda f_H^{(1)}(\gamma)\}
\cdot
f_{H}^*(\gamma)
\\
&\quad{}+
\frac{(1-\rho - \rho^+)\gamma}{\phi(\gamma) - \lambda +  \lambda
f_H^*(\gamma)} \cdot f_H^{(1)}(\gamma)
\\
&=
-
\frac{1-\rho - \rho^+}{\lambda f_H^*(\gamma)} \cdot
f_{H}^*(\gamma)
\\
&\quad{}+
\frac{(1-\rho - \rho^+)\gamma}
{(\lambda f_H^*(\gamma))^2} 
\cdot
\{1-\lambda^+ f_{H^+}^{(1)}(\gamma) - \lambda f_H^{(1)}(\gamma)\}
\cdot
f_{H}^*(\gamma)
\\
&\quad{}+
\frac{(1-\rho - \rho^+)\gamma}{\lambda f_H^*(\gamma)} \cdot f_H^{(1)}(\gamma)
\\
&=
\frac{1-\rho - \rho^+}{\lambda^2 f_H^*(\gamma)}
\cdot
\{-\lambda f_H^*(\gamma) + \gamma -\gamma \lambda^+ f_{H^+}^{(1)}(\gamma)\}.
\end{align*}
Therefore, we have from (\ref{eq:EA-by-A1-A2}), (\ref{eq:f_A1}), and (\ref{eq:f_A2}),
\begin{align*}
\E[A] &= \E[D] + \frac{\rho^+}{\lambda}
\\
&\qquad{}+ 
\frac{1}{\{1 - \lambda^+ f_{H^+}^{(1)}(\gamma)\}\gamma}
\left(
\frac{1-\rho - \rho^+}{\lambda f_H^*(\gamma)}
\cdot
\{-\lambda f_H^*(\gamma) + \gamma -\gamma \lambda^+ f_{H^+}^{(1)}(\gamma)\}
+
1-\rho-\rho^+
\right).
\end{align*}
We thus obtain (\ref{eq:EA}) from this equation and (\ref{eq:ED}).
\qed

\end{document}